\documentclass[11pt,a4paper]{article}

\usepackage[utf8]{inputenc}
\usepackage[T1]{fontenc}
\usepackage{lmodern}

\usepackage{amsmath,amssymb,amsthm}
\usepackage{bm}
\usepackage{mathtools}

\usepackage{graphicx}

\usepackage{booktabs}
\usepackage{array}

\usepackage[margin=2.5cm]{geometry}
\usepackage{setspace}
\usepackage{parskip}

\usepackage[authoryear,round]{natbib}

\usepackage[colorlinks=true,linkcolor=blue,citecolor=blue,urlcolor=blue]{hyperref}
\usepackage{cleveref}

\newtheorem{proposition}{Proposition}

\theoremstyle{definition}
\newtheorem{definition}{Definition}
\theoremstyle{remark}
\newtheorem*{remark}{Remark}

\newcommand{\dd}{\mathrm{d}}
\newcommand{\E}{\mathbb{E}}

\newcommand{\tw}{\tau_w}
\newcommand{\tc}{\tau_c}
\newcommand{\td}{\tau_d}
\newcommand{\tk}{\tau_k}
\newcommand{\rf}{r_f}

\title{Heterogeneous Returns and Wealth Tax Neutrality:\\
A Fokker--Planck Framework}
\author{Anders G Fr{\o}seth\thanks{Independent Researcher.
  E-mail: \href{mailto:indrefjorden@pm.me}{indrefjorden@pm.me}.}}
\date{\today}

\begin{document}
\maketitle

\begin{abstract}
We extend the Fokker--Planck framework of \citet{Froeseth2026S} to
populations of investors with heterogeneous, persistent
return-generating ability.  When the drift coefficient in the Langevin
equation for log-wealth varies across investors, the proportional
wealth tax remains a uniform drift shift but ceases to be neutral in
the economic sense: its real incidence differs across ability types,
and the stationary wealth distribution changes shape.  We derive the
extended Fokker--Planck equation on the joint space of log-wealth and
ability, characterise the conditions under which the drift-shift
symmetry breaks, and identify the consequences for asset prices and
portfolio allocations.  The analysis connects the neutrality results of
\citet{Froeseth2026N} and the Fokker--Planck dynamics of
\citet{Froeseth2026S} to the heterogeneous-returns literature,
notably the ``use-it-or-lose-it'' mechanism of
\citet{GuvenEtAl2023}.
\end{abstract}

\section{Introduction}\label{sec:intro}

The distinction between a capital income tax and a wealth tax can be
stated in two lines.  Let $a_i$ denote the wealth of investor~$i$,
$r_i$ the rate of return on that wealth, and $\tau_k$, $\tau_a$ the
flat tax rates on capital income and wealth respectively.  Under a
capital income tax, the after-tax wealth of investor~$i$ is
\begin{equation}\label{eq:income_tax}
  a_i^{\text{after-tax}} = a_i + (1 - \tau_k) \cdot r_i \, a_i \,,
\end{equation}
whereas under a wealth tax, assessed on the end-of-period market
value, it is
\begin{equation}\label{eq:wealth_tax}
  a_i^{\text{after-tax}} = (1 - \tau_a) \cdot (a_i + r_i \, a_i) \,.
\end{equation}
These two expressions, following \citet{GuvenEtAl2023}, contain the
entire distinction.  The income tax~\eqref{eq:income_tax} scales only
the \emph{return} $r_i a_i$: the tax base is the flow.  The wealth
tax~\eqref{eq:wealth_tax} scales the \emph{entire end-of-period
position} $(1 + r_i) a_i$: the tax base is the stock.\footnote{The
end-of-period convention is used throughout this paper series;
see \citet{Froeseth2026N}, Remark~7.7, for the tax-base timing
discussion and the connection to the \citet{Kruschwitz2023} arbitrage
condition.  In a single period the result is identical to
beginning-of-period assessment, since both yield
$(1-\tau_a)(1+r_i)\,a_i$.}

When the rate of return is common across investors---$r_i = r$ for
all~$i$---the two systems are equivalent up to a rate adjustment,
$\tau_a = r\tau_k/(1 + r)$.  Both reduce every investor's after-tax
wealth by the same fraction.  Because the return is homogeneous, there
is no structural difference between taxing the flow and taxing the
stock.

But when $r_i$ varies across investors, the equivalence breaks.  The
income tax burden is proportional to the individual return: a
high-return investor pays more; a zero-return investor pays nothing.
The wealth tax, by contrast, falls on the entire end-of-period
position: even a zero-return investor pays $\tau_a\, a_i$.  The
consequence, over time, is that wealth migrates from low-return to
high-return investors under the wealth tax---an effect absent under
the income tax.  This is the ``use-it-or-lose-it'' mechanism of
\citet{GuvenEtAl2023}.

The present paper develops this observation in the continuous-time
Fokker--Planck framework of
\citet{Froeseth2026N,Froeseth2026E,Froeseth2026S}.  In that
framework, a proportional wealth tax enters as a uniform reduction of
the drift coefficient in the Langevin equation for log-wealth, leaving
the diffusion structure unchanged.  This drift-shift symmetry preserves
the shape of the wealth distribution and the optimal portfolio weights
of all investors.  A central consequence, developed in
\citet{Froeseth2026R}, is the \emph{redistribution paradox}: because
the drift shift is uniform, a proportional wealth tax does not alter
the shape of the stationary wealth distribution.  The tax is
simultaneously non-distortionary and non-redistributive through the
market channel---any redistribution must come through the fiscal
channel (tax revenue spent on transfers or public goods).

This paradox rests on the assumption that all investors face the same
return process.  Individual wealth paths differ because the Brownian
increments are independent draws from a common distribution, not
because investors differ in their capacity to generate returns.  The
resulting inequality arises from luck within a shared stochastic
process.  Under these conditions, the wealth tax cannot alter relative
positions: it shifts the entire distribution uniformly in log-wealth.

Empirically, this assumption is well supported in deep, liquid public
markets.  The fund performance literature shows little evidence of
persistent outperformance among active managers investing in listed
equities \citep{FamaFrench2010}.  After fees, actively managed funds
do not systematically beat passive benchmarks, consistent with the view
that no investor has a durable edge.

However, in less efficient segments of the capital market---venture
capital, private businesses, and entrepreneurial activity---there is clear
evidence of persistent performance differences across managers and
firms.  \citet{KaplanSchoar2005} find that top-quartile venture
capital funds are significantly more likely to raise top-quartile
successor funds.  \citet{KortewegSorensen2017} confirm persistent
skill net of selection effects.  Most directly,
\citet{FagerengEtAl2020} document substantial persistence in
individual rates of return using 20~years of Norwegian administrative
panel data, finding that individual fixed effects explain a large
share of return variation.

This empirical heterogeneity motivates the present paper.  At the
theoretical level, \citet{BenhabibBisinZhu2011} established that
heterogeneous returns---not just earnings risk---are needed to
generate realistic Pareto tails in the wealth distribution.
\citet{CaoLuo2017} show in a general equilibrium model that persistent
return heterogeneity produces a Pareto tail whose index depends on
equilibrium variables, and that changes in corporate taxation and
financial regulation can account for the joint evolution of rising
wealth inequality and declining labour share.  The continuous-time
heterogeneous-agent framework of \citet{AchdouEtAl2022} provides the
mathematical toolkit---Kolmogorov forward equations on the joint space
of wealth and individual states---that we adopt here.

We ask: what happens to the drift-shift symmetry---and hence to the
neutrality result---when investors differ in a persistent,
individual-specific ability parameter that governs their expected
return?  The question is motivated directly by \citet{GuvenEtAl2023},
who build an overlapping-generations model with heterogeneous
entrepreneurial ability and show that wealth taxation dominates capital
income taxation in efficiency terms.  Their ``use-it-or-lose-it'' mechanism---the
reallocation of capital from low-ability to high-ability
entrepreneurs---is absent in homogeneous-agent models and operates
precisely through the differential real incidence of a uniform tax on
investors with heterogeneous returns.

A central finding of the present analysis is that the common objection
to wealth taxation---that it ``punishes skill''---reverses the actual
mechanism.  Under heterogeneous returns, the wealth tax \emph{preserves}
drift differences between high- and low-ability investors; it is the
income tax that compresses them.  The genuine difficulty is subtler:
the wealth tax cannot distinguish \emph{why} returns are persistent.
Skill and market structure interact multiplicatively---a skilled
entrepreneur earns persistent excess returns only because the market
is too inefficient to arbitrage the advantage away---so the tax
simultaneously rewards productive ability and amplifies structural
rents.

We proceed as follows.  \Cref{sec:recap} recaps the homogeneous
framework from \citet{Froeseth2026S}.  \Cref{sec:hetero} introduces
the extended Fokker--Planck equation on the joint space of log-wealth
and ability, derives the marginal wealth dynamics, and characterises the
conditions under which the drift-shift symmetry breaks.
\Cref{sec:neutrality} analyses the consequences for neutrality,
equilibrium prices, and portfolio allocations.
\Cref{sec:flow_stock} extends the analysis to the combined
flow-and-stock tax system of \citet{Froeseth2026F}, showing that the
distinction between taxing the return (flow) and taxing the level
(stock), which is immaterial under homogeneous returns, becomes the
central question under heterogeneous returns.
\Cref{sec:discussion} discusses the empirical domain of each
assumption and identifies directions for further work.

\section{The homogeneous framework: recap}\label{sec:recap}

We briefly recall the setup from \citet{Froeseth2026S}, establishing
notation for the extension that follows.

\subsection{Individual dynamics}

An investor's wealth $W(t)$ evolves under geometric Brownian motion:
\begin{equation}\label{eq:gbm}
  \frac{\dd W}{W} = \mu \, \dd t + \sigma \, \dd B_t \,,
\end{equation}
where $\mu$ is the expected instantaneous return, $\sigma > 0$ is
volatility, and $B_t$ is a standard Brownian motion.  In log-wealth
$x \equiv \ln W$:
\begin{equation}\label{eq:langevin}
  \dd x = v \, \dd t + \sigma \, \dd B_t \,,
  \qquad v \equiv \mu - \frac{\sigma^2}{2} \,.
\end{equation}
This is a Langevin equation with constant drift velocity~$v$ and
additive noise of strength~$\sigma$.

\subsection{Population dynamics}

For a population of $\mathcal{N}$ investors with common
$(\mu, \sigma)$, the density $\pi(x,t)$ of log-wealth evolves
according to the Fokker--Planck equation:
\begin{equation}\label{eq:fp_hom}
  \frac{\partial \pi}{\partial t}
  = -v \frac{\partial \pi}{\partial x}
    + D \frac{\partial^2 \pi}{\partial x^2} \,,
  \qquad D \equiv \frac{\sigma^2}{2} \,.
\end{equation}

\subsection{The drift-shift symmetry}

A proportional wealth tax at rate $\tw$ modifies only the drift:
\begin{equation}\label{eq:fp_tax}
  \frac{\partial \pi}{\partial t}
  = -(v - \tw) \frac{\partial \pi}{\partial x}
    + D \frac{\partial^2 \pi}{\partial x^2} \,.
\end{equation}
The transformation $v \mapsto v - \tw$ is the drift-shift symmetry.
It preserves the diffusion coefficient, the shape of the propagator,
and---when coupled with the asset pricing analysis of
\citet{Froeseth2026N}—the optimal portfolio weights.

The key structural requirement is that the drift shift is
\emph{uniform}: every investor experiences the same reduction $\tw$ in
drift velocity.  This uniformity is what we now relax.

\section{Heterogeneous returns: the extended framework}\label{sec:hetero}

\subsection{Ability as a second state variable}\label{sec:ability}

We introduce a persistent, individual-specific ability parameter $z$
that governs the investor's expected return.  Specifically, replace
\eqref{eq:gbm} with
\begin{equation}\label{eq:gbm_z}
  \frac{\dd W_i}{W_i} = \mu(z_i) \, \dd t + \sigma(z_i) \, \dd B_t^i \,,
  \qquad
  \mu(z_i) \;\equiv\; \E^{\mathbb{P}}\!\bigl[r_i \,\big|\, z = z_i\bigr] \,,
\end{equation}
where the conditional expectation is under the physical
measure~$\mathbb{P}$: investor~$i$ earns a different expected return
because the return-generating process itself depends on ability, not
because investors hold different beliefs about a common process.  The
volatility $\sigma(z)$ is similarly ability-dependent, and $B_t^i$ are
independent Brownian motions.  We allow both drift and diffusion to depend on ability,
though the case $\sigma(z) = \sigma$ (common volatility, heterogeneous
drift) already captures the essential mechanism.

The ability parameter~$z$ is itself stochastic.  We model it as a
diffusion process:
\begin{equation}\label{eq:z_dynamics}
  \dd z_i = -\gamma(z_i - \bar{z}) \, \dd t
    + \eta \, \dd Z_t^i \,,
\end{equation}
where $\gamma > 0$ controls the rate of mean-reversion toward the
population average $\bar{z}$, $\eta > 0$ is the ability noise
strength, and $Z_t^i$ is a Brownian motion independent of $B_t^i$.
The mean-reversion captures the empirical finding that ability
persistence is real but imperfect: top-performing investors tend to
revert over long horizons, particularly across generations.

\begin{remark}[Interpretation of~$z$]
The ability parameter $z$ can represent several distinct economic
mechanisms: genuine skill in identifying undervalued assets or managing
enterprises; informational advantage from networks, experience, or
sector knowledge; access to deal flow or co-investment opportunities
unavailable to the general market; or scale advantages that larger
portfolios command (lower transaction costs, better terms, access to
illiquid strategies).  The mathematical framework is agnostic about the
source of heterogeneity; what matters is that $z$ is persistent and
affects expected returns.  The empirical evidence
\citep{FagerengEtAl2020} suggests that individual fixed effects in
returns are substantial and long-lived, without cleanly identifying
which mechanism dominates.
\end{remark}

\subsection{The joint Fokker--Planck equation}\label{sec:joint_fp}

Define log-wealth $x_i \equiv \ln W_i$.  By It\^{o}'s lemma:
\begin{equation}\label{eq:langevin_z}
  \dd x_i = v(z_i) \, \dd t + \sigma(z_i) \, \dd B_t^i \,,
  \qquad v(z) \equiv \mu(z) - \frac{\sigma(z)^2}{2} \,.
\end{equation}
The state of each investor is the pair $(x_i, z_i)$.  Let
$f(x, z, t)$ denote the joint density: $f(x,z,t)\, \dd x\, \dd z$ is
the fraction of investors with log-wealth in $[x, x+\dd x]$ and
ability in $[z, z+\dd z]$ at time~$t$.

Since $(x_i, z_i)$ satisfies a two-dimensional It\^{o} diffusion with
independent noise sources, the joint density evolves according to the
Fokker--Planck equation:
\begin{equation}\label{eq:fp_joint}
  \boxed{
  \frac{\partial f}{\partial t}
  = -\frac{\partial}{\partial x}\bigl[v(z)\, f\bigr]
    + \frac{\partial^2}{\partial x^2}\bigl[D(z)\, f\bigr]
    + \frac{\partial}{\partial z}\bigl[\gamma(z - \bar{z})\, f\bigr]
    + \frac{\eta^2}{2}\frac{\partial^2 f}{\partial z^2} \,,}
\end{equation}
where $D(z) \equiv \sigma(z)^2/2$ is the ability-dependent diffusion
coefficient.

The first two terms describe the wealth dynamics at fixed ability: an
investor of ability~$z$ drifts in log-wealth at velocity $v(z)$ and
diffuses with coefficient $D(z)$.  The third and fourth terms describe
the ability dynamics: mean-reversion toward $\bar{z}$ at rate~$\gamma$
and diffusion with coefficient $\eta^2/2$.

\begin{remark}[The homogeneous limit]
When $\mu(z) = \mu$ and $\sigma(z) = \sigma$ for all~$z$, the wealth
dynamics decouple from ability.  Integrating \eqref{eq:fp_joint} over
$z$ recovers the homogeneous Fokker--Planck equation
\eqref{eq:fp_hom}, which is the Fokker--Planck formulation of the
multiplicative wealth dynamics studied by
\citet{BouchaudMezard2000}.  The drift-shift symmetry holds in this limit
because the drift is independent of the state variable~$z$ over which
ability is distributed.
\end{remark}

\subsection{The taxed dynamics}\label{sec:taxed}

Under a proportional wealth tax at rate~$\tw$:
\begin{equation}\label{eq:gbm_z_tax}
  \frac{\dd W_i}{W_i} = (\mu(z_i) - \tw) \, \dd t
    + \sigma(z_i) \, \dd B_t^i \,.
\end{equation}
In log-wealth:
\begin{equation}\label{eq:langevin_z_tax}
  \dd x_i = (v(z_i) - \tw) \, \dd t + \sigma(z_i) \, \dd B_t^i \,.
\end{equation}
The taxed joint Fokker--Planck equation is:
\begin{equation}\label{eq:fp_joint_tax}
  \frac{\partial f}{\partial t}
  = -\frac{\partial}{\partial x}\bigl[(v(z) - \tw)\, f\bigr]
    + \frac{\partial^2}{\partial x^2}\bigl[D(z)\, f\bigr]
    + \frac{\partial}{\partial z}\bigl[\gamma(z - \bar{z})\, f\bigr]
    + \frac{\eta^2}{2}\frac{\partial^2 f}{\partial z^2} \,.
\end{equation}
The tax enters as $v(z) \mapsto v(z) - \tw$, exactly as in the
homogeneous case: a uniform shift of the drift coefficient.

\subsection{Where the symmetry breaks}\label{sec:symmetry_break}

The drift-shift transformation $v(z) \mapsto v(z) - \tw$ is still a
symmetry of the Fokker--Planck \emph{operator}: it changes only the
advection term in the $x$-direction.  But the economic content of
neutrality---that the shape of the wealth distribution and the
relative positions of investors are preserved---no longer follows.

To see why, consider the marginal wealth density
$\pi(x,t) = \int f(x,z,t)\, \dd z$.  Integrating
\eqref{eq:fp_joint_tax} over~$z$:
\begin{equation}\label{eq:fp_marginal}
  \frac{\partial \pi}{\partial t}
  = -\frac{\partial}{\partial x}
    \Bigl[\bigl(\langle v(z) \rangle_x - \tw\bigr)\, \pi\Bigr]
    + \frac{\partial^2}{\partial x^2}
    \bigl[\langle D(z) \rangle_x \, \pi\bigr]
    + \underbrace{\frac{\partial}{\partial z}\bigl[\gamma(z - \bar{z})\, f\bigr]
    + \frac{\eta^2}{2}\frac{\partial^2 f}{\partial z^2}}_{\to\, 0} \,,
\end{equation}
where the $z$-flux terms vanish upon integration over~$z$ with natural
boundary conditions, and $\langle \cdot \rangle_x$ denotes the
conditional expectation over ability given log-wealth~$x$:
\begin{equation}\label{eq:cond_avg}
  \langle v(z) \rangle_x
  \equiv \frac{\int v(z)\, f(x,z,t)\, \dd z}{\pi(x,t)} \,.
\end{equation}
The effective drift in the marginal equation is not a constant but
depends on~$x$ through the conditional average $\langle v(z)
\rangle_x$.  If high-ability investors are concentrated at high wealth
levels (a natural consequence of persistent outperformance), then
$\langle v(z) \rangle_x$ is an increasing function of~$x$.  The
marginal Fokker--Planck equation has \emph{state-dependent drift}: the
rich face a higher effective drift than the poor, not because the tax
treats them differently, but because they are disproportionately
high-ability.

\begin{proposition}[Neutrality requires ability--wealth independence]
  \label{prop:neutrality_condition}
The drift-shift $v(z) \mapsto v(z) - \tw$ preserves the marginal
wealth distribution $\pi(x,t)$ if and only if the conditional average
$\langle v(z) \rangle_x$ is independent of~$x$---that is, ability and
wealth are statistically independent in the joint distribution $f$.
\end{proposition}

\begin{proof}
If $\langle v(z) \rangle_x = \bar{v}$ is independent of~$x$, the
marginal equation reduces to \eqref{eq:fp_hom} with constant drift
$\bar{v}$, and the shift $\bar{v} \mapsto \bar{v} - \tw$ is the
standard drift-shift symmetry.

Conversely, if $\langle v(z) \rangle_x$ depends on~$x$, the marginal
drift is state-dependent and the shift $\langle v(z) \rangle_x
\mapsto \langle v(z) \rangle_x - \tw$ does not preserve the shape of
the stationary distribution: the ratio of drift to diffusion changes
as a function of~$x$, altering the Pareto exponent and the
distributional shape.
\end{proof}

But ability and wealth cannot remain independent when ability is
persistent and affects returns.  High-ability investors accumulate
faster; over time, they migrate to the upper tail.  The joint
distribution $f(x,z,t)$ develops a positive correlation between $x$
and~$z$ even if the initial distribution is factored.

\begin{proposition}[Endogenous correlation]\label{prop:correlation}
Suppose $v(z)$ is strictly increasing and the initial condition is
factored: $f(x,z,0) = \pi_0(x)\, g_0(z)$.  Then for any $t > 0$, the
conditional mean $\langle z \rangle_x$ is a strictly increasing
function of~$x$.  The wealth and ability variables are positively
correlated:
\begin{equation}
  \mathrm{Cov}(x, z) > 0 \qquad \text{for all } t > 0 \,.
\end{equation}
\end{proposition}

\begin{proof}[Proof sketch]
At $t = 0$, all wealth levels share the same ability distribution.
Over the interval $[0, \dd t]$, investors with $z > \bar{z}$ receive
drift $v(z) > \bar{v}$ and shift rightward in log-wealth relative to
the mean; investors with $z < \bar{z}$ shift leftward.  This
differential displacement creates a positive tilt in the conditional
distribution of~$z$ given~$x$: the right tail of log-wealth becomes
enriched in high-$z$ investors, and the left tail in low-$z$
investors.  The mean-reversion in $z$ partially offsets this tilt but
cannot eliminate it as long as $v(z)$ is nonconstant and ability
persistence is nonzero ($\gamma < \infty$).
\end{proof}

\section{Consequences for neutrality, prices, and
  portfolios}\label{sec:neutrality}

\subsection{Neutrality breaking: the mechanism}

Propositions~\ref{prop:neutrality_condition}
and~\ref{prop:correlation} together establish that the drift-shift
symmetry breaks whenever returns are heterogeneous and persistent.
The uniform reduction $v(z) \mapsto v(z) - \tw$ imposes a larger
\emph{relative} burden on low-ability investors, whose gross drift is
smaller: their net drift $v(z) - \tw$ may turn negative while
high-ability investors remain in surplus.  Wealth migrates from
low-$z$ to high-$z$ types, reshaping the stationary distribution.
This is \citeauthor{GuvenEtAl2023}'s ``use-it-or-lose-it'' mechanism
in Fokker--Planck language: the tax does not alter any individual's
stochastic environment (it is non-distortionary in the sense of P1)
but it is not neutral in the distributional sense of P3.

The redistribution paradox of \citet{Froeseth2026R}---that a
proportional wealth tax is simultaneously non-distortionary and
non-redistributive---therefore holds only in the homogeneous limit.
Under heterogeneous returns, the same tax actively redistributes
wealth through the market channel by imposing a uniform cost on
investors with heterogeneous capacities to bear it.

\subsection{The effective Pareto exponent}

In the homogeneous framework, the stationary Pareto exponent of the
wealth distribution (when source and sink terms are included) depends
on the drift-to-diffusion ratio $v/D$.  A uniform drift shift
preserves this ratio's structure and hence the Pareto exponent
\citep{Froeseth2026R}.

With heterogeneous ability, the effective Pareto exponent of the
upper tail is governed by the high-$z$ investors who dominate the
right tail:
\begin{equation}\label{eq:pareto_het}
  \alpha_{\mathrm{eff}} \approx 1
    + \frac{v(z_{\max}) - \tw}{D(z_{\max})} \,,
\end{equation}
where $z_{\max}$ is the ability level of the investors who populate
the tail.  The tax $\tw$ reduces the effective Pareto exponent,
thinning the right tail---a redistributive effect absent in the
homogeneous framework.

\begin{remark}[Dependence on the full ability distribution]
Equation~\eqref{eq:pareto_het} singles out the tail-dominating ability
level $z_{\max}$.  In the mean-field model of
\citet{BernardBouchaudLeDoussal2026}, where agents have quenched
heterogeneous growth rates drawn from a distribution with variance
$\Sigma_0^2$, the stationary wealth tail exponent in the partially
localised phase takes the form $\mu = 1 - \Sigma_0^2/\sigma^4$.
The tail exponent thus depends on the \emph{variance} of the growth-rate
distribution, not only on its maximum.  This suggests that
\eqref{eq:pareto_het} captures the leading-order effect but that the
full ability distribution---in particular the dispersion of
returns---enters at the next order, further thinning or thickening the
tail beyond what the single dominant type predicts.
\end{remark}

\subsection{Asset prices}\label{sec:prices}

In the homogeneous framework, asset prices are invariant under the
wealth tax because all investors are marginal for all assets, and the
drift shift affects them identically.  The stochastic discount factor
is common.

With heterogeneous ability, the set of marginal investors for each
asset may differ.  If the wealth tax causes low-ability investors to
shed risky assets (because they cannot sustain the tax from their lower
returns), the identity of the marginal buyer changes.  The remaining
marginal investors---disproportionately high-ability---have different
risk preferences and opportunity sets, shifting equilibrium risk premia.

\citet{GuvenEtAl2023} find this quantitatively in their calibrated
model: the equilibrium interest rate falls and the return to
entrepreneurial capital rises after a tax reform from capital income
taxation to wealth taxation.  In the Fokker--Planck language, this
corresponds to the aggregate drift $\langle v(z) \rangle_x$ adjusting
endogenously as the composition of investors at each wealth level
changes.

\subsection{Portfolio allocations}

The portfolio weight invariance of \citet{Froeseth2026N} depends on
multiplicative separability: the tax scales all asset payoffs by the
same factor, so the tangent portfolio is unchanged.

When ability interacts with returns, the effective opportunity set
differs across investors.  A high-$z$ investor's after-tax excess
return on risky capital is $\mu(z_{\mathrm{high}}) - \tw - r_f$; a
low-$z$ investor's is $\mu(z_{\mathrm{low}}) - \tw - r_f$.  The ratio
of these excess returns changes with~$\tw$:
\begin{equation}\label{eq:ratio}
  \frac{\mu(z_{\mathrm{high}}) - \tw - r_f}
       {\mu(z_{\mathrm{low}}) - \tw - r_f}
\end{equation}
is not invariant under shifts in~$\tw$ (it is undefined when the
denominator passes through zero).  Investors optimally reallocate:
low-ability investors exit risky assets; high-ability investors absorb
the freed capacity.  Portfolio invariance breaks not because the tax is
non-proportional, but because investors face different effective
opportunity sets when ability interacts with the tax.

\section{Flow Taxes Versus Stock Taxes Under Heterogeneous
  Returns}\label{sec:flow_stock}

The preceding analysis focused on the wealth tax (a stock tax) in
isolation.  In practice, investors face both stock and flow taxes
simultaneously.  \citet{Froeseth2026F} shows that under three
structural conditions---equal capital income and corporate tax rates
(C1), shielding at the risk-free rate (C2), and uniform wealth tax
assessment (C3)---the combined system acts through a
\emph{drift-shift-and-rescale} transformation that preserves portfolio
neutrality.  The flow taxes rescale excess drifts by
$(1 - \tc)(1 - \td)$; the wealth tax shifts all drifts by $-\tw$.
Neither modification alters relative drifts between assets.

Ability heterogeneity breaks this symmetry in structurally different
ways for the two tax types.

\subsection{A motivating example}\label{sec:example}

The following example, adapted from \citet{GuvenEtAl2023}, illustrates
the asymmetry.

Consider two investors, each with wealth $a = 1000$.  Investor~1 earns
return $r_1 = 0\%$ (low ability); investor~2 earns $r_2 = 20\%$ (high
ability).  Compare a capital income tax at rate $\tk = 25\%$ with a
wealth tax at rate $\tw = 2.5\%$, assessed on end-of-period wealth
(consistent with equation~\eqref{eq:wealth_tax}).

\begin{table}[ht]
\centering
\caption{After-tax wealth under capital income tax versus wealth tax,
  for two investors with identical wealth but different returns.}
\label{tab:example}
\small
\begin{tabular}{@{}lccc@{}}
\toprule
& \textbf{Pre-tax} & \textbf{Cap.\ income tax}
  & \textbf{Wealth tax} \\
& & $\tk = 25\%$ & $\tw = 2.5\%$ \\
\midrule
\textit{Investor 1} ($r_1 = 0\%$) &&& \\[2pt]
\quad End-of-period wealth & 1000 & 1000 & 1000 \\
\quad Tax paid & --- & 0 & 25 \\
\quad After-tax wealth & 1000 & 1000 & 975 \\[6pt]
\textit{Investor 2} ($r_2 = 20\%$) &&& \\[2pt]
\quad End-of-period wealth & 1200 & 1200 & 1200 \\
\quad Tax paid & --- & 50 & 30 \\
\quad After-tax wealth & 1200 & 1150 & 1170 \\[6pt]
\midrule
Total tax revenue & --- & 50 & 55 \\
Wealth gap (after tax) & 200 & 150 & 195 \\
\bottomrule
\end{tabular}
\end{table}

Under the capital income tax, investor~1 pays nothing (no income to
tax) and investor~2 bears the entire burden.  The wealth gap narrows
from 200 to 150: the flow tax compresses the distribution by taxing
the productive investor more heavily.

Under the wealth tax, assessed on end-of-period wealth, investor~1
pays $0.025 \times 1000 = 25$ and investor~2 pays
$0.025 \times 1200 = 30$.  The tax is not identical---investor~2's
higher return enlarges the tax base---but the compression is marginal:
the wealth gap falls from 200 to 195, compared with 150 under the
income tax.  The stock tax nearly preserves relative positions while
the flow tax compresses them sharply.  And the \emph{composition}
changes: investor~1 loses wealth in absolute terms (from 1000 to 975),
while investor~2 gains (from 1000 to 1170, net of tax).  Over multiple
periods, this asymmetry compounds: capital migrates from low-return to
high-return investors.  (Under homogeneous returns, $r_1 = r_2$, the
two systems are equivalent---cf.\ \Cref{sec:intro}.)

\subsection{The Fokker--Planck formulation}\label{sec:fp_flow_stock}

In the continuous-time framework, consider the
combined tax system of \citet{Froeseth2026F}: corporate tax at
rate~$\tc$, capital income tax at rate~$\tk$, dividend/gains tax at
rate~$\td$, and wealth tax at rate~$\tw$.  For an investor of
ability~$z$ holding a portfolio of risky assets in corporate form and a
risk-free asset held personally, the after-tax dynamics in log-wealth
are:
\begin{equation}\label{eq:combined_z}
  \dd x_i = \Bigl[
    \underbrace{(1 - \tc)(1 - \td)\bigl(\mu(z_i) - \rf\bigr)}_{
      \text{rescaled excess drift (flow taxes)}}
    + \underbrace{\rf(1 - \tk)}_{
      \text{after-tax risk-free}}
    - \underbrace{\tw\vphantom{\rf(1)}}_{
      \text{drift shift (stock tax)}}
    - \tfrac{\sigma(z_i)^2}{2}
  \Bigr] \dd t
  + \sigma(z_i) \, \dd B_t^i \,.
\end{equation}

The flow taxes enter through the factor $(1 - \tc)(1 - \td)$, which
multiplies the \emph{excess return} $\mu(z_i) - \rf$.  The wealth tax
enters as $-\tw$, a uniform shift independent of returns.

Under heterogeneous returns, these two modifications have structurally
different incidence:

\textbf{Flow tax incidence.}  The flow-tax burden on investor~$i$ is
proportional to the excess return $\mu(z_i) - \rf$.  A high-ability
investor earns a high excess return and pays correspondingly more flow
tax.  A low-ability investor earns a low (or zero) excess return and
pays correspondingly less.  The flow tax is \emph{ability-weighted}:
its burden scales with~$z$.

\textbf{Stock tax incidence.}  The wealth-tax burden on investor~$i$
is $\tw$ per unit of log-wealth, independent of~$z$.  Both high- and
low-ability investors pay the same tax per unit of wealth.  But the
\emph{real} burden---the tax relative to the investor's capacity to
regenerate wealth---is inversely related to ability.  The stock tax
is \emph{ability-blind} in its rate but \emph{ability-dependent} in
its real incidence.

\subsection{Drift decomposition under heterogeneous
  returns}\label{sec:drift_decomp}

Write the after-tax drift of investor~$i$ as:
\begin{equation}\label{eq:drift_decomp}
  v^{\mathrm{tax}}(z_i)
  = \underbrace{(1 - \tc)(1 - \td)}_{\equiv\, \lambda}\,
    \bigl(\mu(z_i) - \rf\bigr)
    + \rf(1 - \tk) - \tw - \tfrac{\sigma(z_i)^2}{2} \,.
\end{equation}
Consider two investors with abilities $z_H > z_L$.  Their drift
difference is:
\begin{equation}\label{eq:drift_diff}
  v^{\mathrm{tax}}(z_H) - v^{\mathrm{tax}}(z_L)
  = \lambda\bigl[\mu(z_H) - \mu(z_L)\bigr]
    - \tfrac{1}{2}\bigl[\sigma(z_H)^2 - \sigma(z_L)^2\bigr] \,.
\end{equation}
The wealth tax $\tw$ does not appear: it cancels in the drift
\emph{difference}.  The flow taxes appear through $\lambda < 1$: they
\emph{compress} the drift difference between the two investors.

This reveals the fundamental asymmetry:

\begin{proposition}[Differential incidence of flow and stock taxes]
\label{prop:differential}
Under heterogeneous returns with $\mu(z_H) > \mu(z_L)$ and common
volatility $\sigma(z) = \sigma$:
\begin{enumerate}
  \item The flow taxes reduce the drift difference between high- and
    low-ability investors by the factor $\lambda = (1 - \tc)(1 - \td)$.
    They compress the ability-driven wealth divergence.
  \item The wealth tax does not affect the drift difference.  It shifts
    all drifts uniformly, preserving the ability-driven divergence rate.
\end{enumerate}
Consequently, for a given total tax burden, a revenue-neutral shift
from flow taxes toward wealth taxation \emph{widens} the drift gap
between high- and low-ability investors, accelerating the reallocation
of wealth from low-$z$ to high-$z$ types.
\end{proposition}

\begin{proof}
Part~(1) follows directly from \eqref{eq:drift_diff}: the flow-tax
factor $\lambda$ multiplies the return difference.  Part~(2) follows
from the cancellation of $\tw$ in \eqref{eq:drift_diff}.
Revenue-neutrality implies that lowering $\tc, \td$ (raising
$\lambda$ toward~1) requires raising $\tw$.  The drift difference
$\lambda[\mu(z_H) - \mu(z_L)]$ increases while the uniform shift
$\tw$ increases but does not enter the difference.
\end{proof}

\begin{remark}[Guvenen's three channels, reinterpreted]
\citet{GuvenEtAl2023} decompose their welfare result into three
channels: the ``use-it-or-lose-it'' reallocation, the general
equilibrium price response, and the behavioral savings response.
Proposition~\ref{prop:differential} provides the Fokker--Planck
counterpart of the first channel: the reallocation arises because the
stock tax preserves drift differences while the flow tax compresses
them.  A revenue-neutral switch from flow to stock taxation therefore
mechanically widens the drift gap, driving the redistribution of
wealth toward high-ability investors that constitutes the
use-it-or-lose-it effect.
\end{remark}

\begin{figure}[tp]
  \centering
  \includegraphics[width=0.85\textwidth]{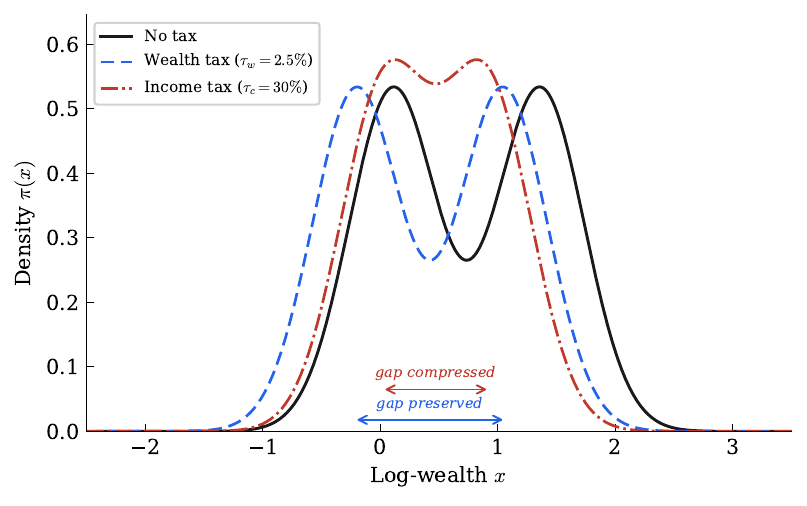}
  \caption{Stationary wealth distribution for a two-type economy
    (high-ability $\mu_H = 12\%$, low-ability $\mu_L = 2\%$) under
    three tax regimes.  The wealth tax shifts both peaks leftward by
    the same amount, preserving the inter-type gap
    (Proposition~\ref{prop:differential}).  The income tax compresses
    the gap by the factor $\lambda = 1 - \tc$.  Parameters:
    $\sigma = 0.15$, $\delta = 0.08$, $\tw = 2.5\%$, $\tc = 30\%$,
    equal population shares.}
  \label{fig:wealth_dist}
\end{figure}

\subsection{Implications for the generalised neutrality
  theorem}\label{sec:gen_neutrality}

The generalised neutrality theorem of \citet{Froeseth2026F} states
that under (C1)--(C3), the combined flow-and-stock tax system
preserves portfolio weights.  The theorem is derived under homogeneous
returns.  How does it extend to heterogeneous returns?

The portfolio optimality condition for investor~$i$ depends on the
after-tax excess returns \emph{available to that investor}.  Under
homogeneous returns, all investors face the same excess returns
(scaled by $\lambda$), and the tangent portfolio is common.  Under
heterogeneous returns, investor~$i$'s after-tax excess return on risky
asset~$k$ relative to the risk-free asset is:
\begin{equation}\label{eq:excess_het}
  R_k^{\mathrm{ex}}(z_i)
  = \lambda\bigl[\mu_k(z_i) - \rf\bigr] - \tw(\alpha_k - \alpha_0) \,,
  \qquad
  \mu_k(z_i) \;\equiv\; \E^{\mathbb{P}}\!\bigl[r_k \,\big|\, z = z_i\bigr] \,.
\end{equation}
As in~\eqref{eq:gbm_z}, the conditional expectation is under the
physical measure~$\mathbb{P}$: heterogeneity enters through
conditioning on the state variable~$z$, not through the probability
measure (cf.\ Remark~\ref{rem:beliefs}).

Under condition~(C3), the wealth tax term cancels across risky assets.
But the return $\mu_k(z_i)$ itself depends on investor ability.  If
high-ability investors earn higher returns on the \emph{same}
assets---through better deal terms, lower transaction costs, or
superior selection---then different investors face different effective
opportunity sets even when holding the same assets.

In the limiting case where ability affects the level of returns
uniformly ($\mu_k(z) = z \cdot \mu_k$ for all~$k$), the relative
excess returns between risky assets are independent of~$z$:
\begin{equation}
  \frac{R_k^{\mathrm{ex}}(z)}{R_j^{\mathrm{ex}}(z)}
  = \frac{\mu_k - \rf/z}{\mu_j - \rf/z} \,,
\end{equation}
which does depend on~$z$ through the effective risk-free rate
$\rf/z$.  Portfolio weights differ across ability types even under
(C1)--(C3).  The generalised neutrality theorem holds
\emph{within} each ability type (tax rates do not affect that type's
portfolio) but the aggregate portfolio---the wealth-weighted average
across types---shifts as the wealth distribution across types changes.

\section{Discussion}\label{sec:discussion}

\subsection{Relationship to the existing framework}

The results above do not invalidate the neutrality framework of
P1--P3.  Rather, they delineate its domain of validity.  The
drift-shift symmetry is a mathematical identity: a proportional wealth
tax always enters as a uniform drift reduction.  The economic content
of neutrality---that the wealth distribution and relative prices are
preserved---holds when this uniform shift translates into a uniform
real burden across investors, which requires that ability and wealth be
statistically independent
(Proposition~\ref{prop:neutrality_condition}).  When ability is
persistent, Proposition~\ref{prop:correlation} shows this
independence cannot hold, and the symmetry breaks economically.

\subsection{What the Fokker--Planck formulation adds}

The joint Fokker--Planck equation \eqref{eq:fp_joint} nests both the
homogeneous and heterogeneous cases.  The deviation from neutrality is
controlled by a single structural feature: $\partial \langle v(z)
\rangle_x / \partial x$, the conditional dependence of drift on
wealth.  This suggests a natural perturbation theory: when ability
heterogeneity is small (low~$\eta$, fast mean-reversion~$\gamma$),
deviations from neutrality are second-order; when heterogeneity is
large, the full joint dynamics must be tracked.

\subsection{Empirical implications}

The framework generates a testable prediction: the degree of
neutrality breaking should correlate with the share of wealth held in
asset classes where return persistence is empirically documented.
Within any economy, the listed-equity and government-bond segments---where
informational efficiency constrains return dispersion---should
exhibit approximate neutrality.  The private-business, venture-capital,
and entrepreneurial segments---where persistent ability differentials
are well documented---should exhibit stronger distributional effects
from wealth taxation.  The aggregate deviation from neutrality depends
on the wealth-weighted mix of these segments.

\paragraph{Norwegian evidence (\citealp{FagerengEtAl2020}).}  \citeauthor{FagerengEtAl2020} use
Norwegian administrative panel data spanning 1993--2013---covering all
taxpayers, with third-party-reported holdings across listed equities,
bonds, deposits, private businesses, and real estate---to establish
three facts directly relevant to the present framework.

First, individual-level return heterogeneity is large: the standard
deviation of real returns on net worth is 8.6 percentage points, and
the unweighted 90th--10th percentile spread exceeds 15 percentage points.
Second, the heterogeneity is persistent: individual fixed effects
account for approximately 26\% of the variance in returns (the
$R^2$ of the return regression rises from 0.30 without individual
fixed effects to 0.50 with them), and the distribution of
these fixed effects has a 90--10 percentile range of nearly
8 percentage points.  Even deposit accounts---the most homogeneous asset
class---exhibit a fixed-effect standard deviation of 2.6 percentage
points, indicating that return heterogeneity extends beyond portfolio
choice into deposit pricing and financial literacy.

Third, returns are positively correlated
with wealth: moving from the 10th to the 90th percentile of net worth
increases the average before-tax return by 11 percentage points,
conditional on individual fixed effects and time effects (the
unconditional gradient is approximately 18 percentage points).

The composition of Norwegian wealth connects these facts to the
present framework.  \citet{FagerengEtAl2020} document that at the top
of the wealth distribution, unlisted business
equity\footnote{Fagereng et al.\ use the term ``private equity'' to
  denote direct ownership stakes in unlisted firms---closely held
  businesses, family firms, and entrepreneurial ventures---as recorded
  in the Norwegian shareholder registry.  This is distinct from the
  financial-industry usage of ``private equity'' to denote buyout and
  venture capital funds.} dominates: it constitutes approximately
38\% of gross wealth for the top 1\% and 85\% for the top
0.01\%.%
\footnote{Computed from Table~1A of \citet{FagerengEtAl2020}, which
  reports gross-wealth composition by fractile.  The 85\% figure is
  read directly from the top-0.01\% row; the 38\% is the approximate
  population-weighted average across the three fractiles comprising
  the top~1\%.}  Return heterogeneity exists across all asset
classes---housing, listed equities, bonds---but is most extreme in
unlisted businesses, where the standard deviation of returns reaches 52
percentage points and business owners exhibit a far wider distribution
of individual fixed effects than non-owners.  The aggregate return
persistence is therefore disproportionately driven by the asset class
that dominates top-percentile wealth.

This interacts with the framework in two ways.  First, the large
private-business share at the top implies a large effective variance
of $\mu(z)$ across wealthy investors, amplifying the neutrality
breaking identified in Propositions~\ref{prop:neutrality_condition}
and~\ref{prop:correlation}.  Second, the wealth composition interacts
with condition~(C3) of \citet{Froeseth2026F}---the requirement that
the wealth tax assessment base is uniform across assets.  The
Norwegian \emph{verdsettelsesrabatt} (valuation discount) assesses
listed equities and unlisted shares at different fractions of market
value.  This breaks (C3) precisely where the return heterogeneity is
concentrated: in the private-business sector.  A third channel reinforces the other two:
\citet{AlstadsaeterEtAl2019} show that tax evasion through offshore
holdings is heavily concentrated at the very top of the wealth
distribution, further eroding the effective uniformity of the
assessment base.  The three symmetry-breaking channels---heterogeneous
returns, non-uniform assessment, and differential evasion---compound
rather than offset.

\paragraph{US evidence (\citealp{SmithZidarZwick2023}).}  \citeauthor{SmithZidarZwick2023} use
administrative tax data to estimate top wealth in the United States
under heterogeneous returns.  They find substantial return
heterogeneity within asset classes---the interest rate on fixed income
at the top is approximately 3.5 times higher than the average---and
document that the top~1\%, 0.1\%, and 0.01\% wealth shares reached
33.7\%, 15.7\%, and 7.1\% respectively by 2016.
\citet{HubmerKrusellSmith2021} decompose the drivers of rising US
wealth inequality and find that portfolio heterogeneity and the
positive correlation between returns and wealth are essential to match
the data---a quantitative confirmation of the endogenous correlation
identified in Proposition~\ref{prop:correlation}.

\paragraph{Cross-country variation.}  The prediction is not that some
countries exhibit neutrality and others do not, but that the degree of
deviation varies with the composition of wealth at the top.  The
Norwegian data show that even in a Nordic economy with well-developed
public equity markets, private-business wealth dominates at the very
top---and it is this segment that drives the aggregate return
persistence.  \citet{GuvenEtAl2023} note that despite much lower
income inequality, Norwegian wealth inequality is comparable to the
United States, consistent with the role of heterogeneous returns in
shaping the wealth distribution.  Countries where top-percentile
wealth is more concentrated in listed equities and diversified
portfolios should lie closer to the homogeneous benchmark.  The
aggregate deviation from neutrality depends on this composition and
varies across economies and over time.

\subsection{Sources of return persistence and the interpretation of
  ability}\label{sec:winner}

The efficiency argument for wealth taxation in \citet{GuvenEtAl2023}
rests on a specific interpretation of the ability parameter~$z$: it
reflects \emph{productive skill} that generates genuine economic
value.  A high-$z$ investor earns high returns because she selects
better projects, manages firms more effectively, or deploys capital
to higher-value uses.  The use-it-or-lose-it mechanism is then
welfare-improving: the wealth tax accelerates the reallocation of
capital from less productive to more productive hands.

But the empirical finding of persistent return heterogeneity---the
individual fixed effects of \citet{FagerengEtAl2020}---is consistent
with several structural mechanisms that have little to do with
individual skill.  The welfare implications of wealth taxation depend
critically on which mechanism dominates.  We identify three distinct
sources of return persistence, each with different implications for
the framework.

\begin{remark}[Relationship to heterogeneous beliefs]\label{rem:beliefs}
A large literature models return heterogeneity through heterogeneous
beliefs: investors hold different subjective probability measures over
a \emph{common} return-generating process
\citep{Lintner1969,JouiniNapp2006,Basak2005}.  The ability
parameter~$z$ in the present framework is formally more general: each
investor faces a \emph{genuinely different} drift~$\mu(z_i)$, not a
different perception of a common drift.  At the level of the
Fokker--Planck equation, the distinction is invisible---Propositions~\ref{prop:neutrality_condition}--\ref{prop:differential}
hold regardless of the micro-foundation.  It matters, however, for
three reasons: \emph{persistence} (heterogeneous beliefs are
self-correcting under learning; the structural sources identified
below are persistent by construction); \emph{welfare} (compressing
belief-driven drift gaps may protect overconfident investors, while
compressing ability-driven gaps destroys allocative efficiency); and
\emph{calibration} ($\gamma$ should be fast under learning and slow
under structural persistence---the decade-scale persistence documented
by \citet{FagerengEtAl2020} favours the latter).  The framework
encompasses heterogeneous beliefs as a special case but primarily
targets the structural sources catalogued in Categories~I--III below.
\end{remark}

\begin{remark}[Category~I: Skill in private markets]
\label{rem:cat_I_skill}
Some entrepreneurs are genuinely better at deploying capital: better
judgement in project selection, superior management, deeper sector
expertise.  These skill differentials generate persistent return
heterogeneity because private markets lack the arbitrage
mechanisms---short-selling, free capital flow, public information---that
enforce return homogeneity in listed equities (\Cref{sec:intro}).
Under this interpretation, $z$ captures productive ability, and the
use-it-or-lose-it mechanism is efficient.  But a subtle point is
already visible: skill generates persistent returns \emph{only
because} the market is inefficient.  The same entrepreneur in a
perfectly efficient market earns zero excess return.  The persistence
is a joint product of skill and market structure; we formalise this
interaction below.
\end{remark}

\begin{remark}[Category~II: Structural advantage in classical private markets]
\label{rem:cat_II_structure}
Return persistence in private markets can also arise from features of
market structure that are orthogonal to skill.  An investor locked
into a specific firm inherits that firm's trajectory---its sector
exposure, its vintage, its local market position---regardless of her
managerial ability.  Scale economies in capital access mean that
larger firms face lower borrowing costs and better terms, compounding
a size advantage that originated in timing or luck rather than skill.
Survivorship selection inflates apparent persistence: firms that drew
unfavourable shocks exit the sample, and among survivors, returns
appear more persistent than in the full population.  Valuation noise
in administrative data---the Norwegian tax valuations of private
businesses may systematically diverge from market values---can further
amplify measured return heterogeneity in precisely the asset class
where it is already largest.

Under this interpretation, the individual fixed effects of
\citet{FagerengEtAl2020} partly reflect firm characteristics (sector,
size, vintage, local market structure) rather than individual skill.
The wealth tax's use-it-or-lose-it mechanism then redistributes
capital based on structural position rather than productive ability---a
reallocation that is at best neutral and at worst reinforcing of
incumbency advantages.
\end{remark}

\begin{remark}[Category~III: Winner-take-all dynamics in digital markets]
\label{rem:cat_III_winner_take_all}
\citet{Rosen1981} showed that when output can be replicated at low
cost, small quality differences translate into large return
differences.  Digital technologies amplify this through near-zero
marginal cost, global distribution, and network effects
\citep{BrynjolfssonMcAfee2014}.  The empirical consequences are
documented by \citet{AutorEtAl2020} (superstar firms account for
declining labour share) and \citet{DeLoeckerEtAl2020} (rising markups
concentrated among the largest firms).  In a winner-take-all market,
two founders of identical ability launching identical products months
apart can earn vastly different returns: the first mover captures the
market; the second gets nothing.  The return to the dominant firm
reflects structural incumbency---network lock-in, data advantages,
switching costs---not ongoing skill.
\end{remark}

\begin{definition}[Three-way additive decomposition of return persistence]
\label{def:z_decomp}
As a first approximation, the framework's ability parameter $z$
conflates contributions from all three sources:
\begin{equation}\label{eq:z_decomp}
  z_i = \underbrace{z_i^{\mathrm{skill}}}_{\text{I: productive ability}}
      + \underbrace{z_i^{\mathrm{structure}}}_{\text{II: market frictions}}
      + \underbrace{z_i^{\mathrm{platform}}}_{\text{III: winner-take-all rents}} \,.
\end{equation}
\end{definition}

\begin{definition}[Skill--risk decomposition and the $\varphi$ parameter]
\label{def:phi_skill_risk}
A further refinement is needed within \Cref{rem:cat_I_skill}: not all
persistent individual-level return heterogeneity reflects productive
ability---some reflects compensation for undiversifiable risk.  The fund performance literature makes this distinction precise:
after controlling for systematic risk factors, the residual alpha is
small in listed equities (\Cref{sec:intro}) but significant and
persistent for top-quartile venture capital.  The skill component
therefore admits a prior decomposition:
\begin{equation}\label{eq:alpha_risk}
  z_i^{\mathrm{skill}} = z_i^{\alpha} + z_i^{\mathrm{risk}} \,,
\end{equation}
where $z_i^{\alpha}$ is genuine alpha (productive ability net of risk
compensation) and $z_i^{\mathrm{risk}}$ is the return premium for
bearing undiversifiable risk---illiquidity, concentration, vintage
exposure.  Define $\varphi \in [0,1]$ as the fraction of observed
skill-component heterogeneity attributable to risk compensation.  At
$\varphi = 0$, all persistence reflects productive ability and the
\citet{GuvenEtAl2023} efficiency argument applies in full.  At
$\varphi = 1$, all persistence reflects risk bearing: the
use-it-or-lose-it mechanism reallocates capital toward investors who
are simply more exposed to compensated risk, with no efficiency gain.
The welfare case for wealth taxation over income taxation scales with
$(1 - \varphi)$, making the empirical decomposition of return
persistence into alpha and risk-factor components a first-order
question for tax design.
\end{definition}

The five-component framework---$z^{\alpha}$, $z^{\mathrm{risk}}$,
$z^{\mathrm{structure}}$, and $z^{\mathrm{platform}}$ (with the
structural component further split into rents and competitive
frictions)---is developed in the companion literature review.  For the
formal results of \Cref{sec:hetero,sec:neutrality,sec:flow_stock},
the coarser three-way decomposition~\eqref{eq:z_decomp} suffices;
the $\varphi$ parameterisation enters through the welfare
interpretation in \Cref{sec:policy} below.

\begin{definition}[Multiplicative skill--structure interaction model]
\label{def:z_interaction}
The additive decomposition~\eqref{eq:z_decomp} is pedagogically useful
but misleading: as noted under \Cref{rem:cat_I_skill}, skill and
market structure are not independent.  A more faithful representation
is multiplicative:
\begin{equation}\label{eq:z_interaction}
  z_i = \theta_i \cdot h(E_j) + z_j^{\mathrm{structure}} + z_j^{\mathrm{platform}} \,,
\end{equation}
where $\theta_i$ denotes individual ability, $E_j$ measures the
informational efficiency of the market segment~$j$ in which
investor~$i$ operates, and $h(\cdot)$ is decreasing in~$E_j$ with
$h(E) \to 0$ as $E \to 1$ (a fully efficient market).  The first
term captures the interaction: skill earns persistent excess returns
only to the extent that the market fails to compete them away.  The
second and third terms capture rents attributable purely to market
position, independent of who occupies that position.
\end{definition}

The non-separability is empirically grounded.  When high-performing
professionals change institutional context, measured performance drops
sharply, consistent with a large $h(E_j)$ amplifier.  The
\citet{AbowdKramarzMargolis1999} decomposition of wages into person
effects, firm effects, and match-quality residuals confirms that the
interaction term is first-order, not a perturbation.

This interaction has a natural interpretation in the Fokker--Planck
framework.  The mean-reversion rate~$\gamma$ in the
Ornstein--Uhlenbeck dynamics~\eqref{eq:z_dynamics} is not a property
of the individual but of the market segment.  In efficient markets
$\gamma$ is large: skill advantages are quickly arbitraged.  In
private markets $\gamma$ is small: the same skill differentials
persist because no competitive mechanism erodes them.  The
\emph{observed} persistence of the skill component is
$\theta_i / \gamma_j$---both the numerator (ability) and the
denominator (market efficiency) matter.

This non-separability has three consequences.

First, the persistence parameter $\gamma$ differs systematically
across the three categories, but for interacting reasons.  Under
Category~I, $\gamma$ is moderate---not because skill itself
mean-reverts, but because the partial arbitrage mechanisms that do
exist (imitation, entry, technology diffusion) gradually erode the
advantage.  Under Category~II, $\gamma$ is small: structural
advantages---local monopolies, relationship capital, regulatory
protections---face fewer erosion mechanisms.  Under Category~III,
$\gamma^{\mathrm{platform}} \approx 0$: once a platform achieves
dominance, the dynamics are self-reinforcing.  The spectral gap of
the joint Fokker--Planck operator shrinks across the three categories,
and the convergence time of the wealth distribution to its new steady
state lengthens correspondingly.

Second, the welfare criterion changes.  When the skill-times-efficiency
interaction~\eqref{eq:z_interaction} dominates, the planner faces
a genuine dilemma: the use-it-or-lose-it mechanism simultaneously
reallocates capital toward genuine productive ability \emph{and}
toward the market segments where structural inefficiency amplifies
that ability into persistent rents.  When structural and platform
rents dominate (the second and third terms), the planner prefers to
compress drift heterogeneity---which, per
Proposition~\ref{prop:differential}, favours flow taxation over stock
taxation.  The optimal tax design therefore depends on the
decomposition of observed return persistence into its components, a
question that is ultimately empirical and likely sector-specific.

Third, the identification problem is genuine.  An empirical
decomposition exploiting the Norwegian administrative data---regressing
individual return fixed effects on firm-level characteristics (sector,
size, age, market concentration)---separates the individual residual
from the structural component, but the individual residual itself
reflects skill \emph{as amplified by} market inefficiency.  A
restaurant owner in Troms\o{} earning persistent 15\% returns may
be both a talented operator and the beneficiary of a two-competitor
local market; these contributions cannot be separated by conditioning
on firm characteristics alone.  The \citet{AbowdKramarzMargolis1999}
person-firm decomposition is more informative: it isolates the
portable component of returns by tracking individuals across firms.
An entrepreneur whose excess returns survive a move to a different
firm (controlling for firm fixed effects) is demonstrating skill that
is at least partially separable from market structure.  But even the
person fixed effect is contaminated if the individual selects into
equally inefficient markets---the portability test requires moves
across markets of varying efficiency, not just across firms.

\subsection{Policy implications: reinterpreting Guvenen}\label{sec:policy}

A common objection to wealth taxation is that it ``punishes skill'':
by taxing the accumulated capital of successful entrepreneurs, it
penalises precisely those who have demonstrated the ability to deploy
capital productively.  This framing reverses the actual mechanism.
Proposition~\ref{prop:differential} shows that the wealth tax
\emph{preserves} skill-driven drift differences---it is the income
tax that compresses them.  Under heterogeneous returns, a
revenue-neutral shift \emph{from} income taxation \emph{to} wealth
taxation widens the drift gap between high- and low-ability investors,
accelerating the reallocation of capital toward the skilled.  This is
\citeauthor{GuvenEtAl2023}'s central finding: a revenue-neutral switch
from capital income taxation to wealth taxation delivers an average
welfare gain of $\sim 7\%$ of consumption-equivalent for newborns,
with the optimal wealth-tax exercise returning
$\tau_a^{\mathrm{OWT}} \approx 3\%$ (positive, optimal) versus an
optimal capital-income tax of $\tau_k^{\mathrm{OKIT}} = -13.6\%$ (a
subsidy, optimal). Critically, when return heterogeneity is shut down
in their robustness experiment~VI ($z_{ih} = 1$), the optimal
capital-income tax flips from a $-13.6\%$ subsidy to a $+25.1\%$
positive tax---the entire case for wealth taxation over capital income
taxation rests on the heterogeneous-returns channel. The theoretical
companion \citet{GuvenEtAl2024} develops an analytically tractable
infinite-horizon model that establishes the use-it-or-lose-it channel,
characterises the optimal mix of wealth and capital-income taxes, and
adds an innovation/effort channel through which the wealth tax further
increases welfare by incentivising high-productivity entrepreneurial
activity.

There is a genuine paradox in this result.  High-ability investors
accumulate more wealth and therefore pay more wealth tax in absolute
terms ($\tau_w W_i$ is larger when $W_i$ is larger)---the observation
that grounds the ``punishes skill'' objection.  But the uniform
reduction $v(z) \mapsto v(z) - \tau_w$ imposes a larger
\emph{relative} burden on low-ability investors
(\Cref{sec:neutrality}), so wealth migrates toward the skilled over
time.  The absolute burden is proportional to the \emph{stock} of
wealth; the incentive structure depends on the \emph{cross-sectional
distribution} of after-tax drifts.  The tax that \emph{appears} to
burden the wealthy---because it is levied on accumulated wealth---is
the one that preserves the reallocation mechanism.  The income tax, by
contrast, may extract nothing from a zero-return investor, but it
compresses the drift gap: the tax that appears gentler to skilled
investors is the one that erodes the reward to skill.

The distinction between the two tax instruments connects to a deeper
asymmetry in the ``silent partner'' metaphor.
\citet{DomarMusgrave1944} and \citet{Stiglitz1969} established that a
proportional income tax with full loss offset acts as a silent partner
on the \emph{return margin}: the government shares proportionally in
both gains and losses on $r_i - r_f$, and the investor responds by
scaling up risk exposure to restore the pre-tax Sharpe ratio.  The
partnership operates through the return distribution.  The wealth tax,
as formalised in \citet{Froeseth2026N}, acts as a silent partner on
the \emph{position margin}: the government takes a proportional slice
of the asset stock itself, reducing the scale of the position by
$(1 - \tw)$ each period while leaving the return distribution per unit
of remaining wealth unchanged.  The income-tax partnership encourages
more risk-taking; the wealth-tax partnership preserves the risk
composition.  When returns reflect risk compensation (the
$z^{\mathrm{risk}}$ component of~\eqref{eq:alpha_risk}), the two
mechanisms diverge sharply: the income tax triggers the
Domar--Musgrave offset, amplifying gross risk exposure; the wealth tax
leaves the portfolio structure intact.  When returns reflect
productive ability ($z^{\alpha}$), both instruments preserve the
Sharpe ratio, but only the wealth tax preserves the cross-sectional
drift gap (Proposition~\ref{prop:differential}).  The $\varphi$
parameterisation thus determines not only the welfare sign of wealth
taxation but also the behavioural channel through which the two
instruments differ.

The genuine difficulty is not that the wealth tax punishes skill, but
that it cannot distinguish skill from other sources of return
persistence.  The drift gap $\mu(z_H) - \mu(z_L)$ is preserved
regardless of its source (\eqref{eq:z_interaction}):

\begin{enumerate}
  \item Under Category~I (skill amplified by market inefficiency),
    the wealth tax accelerates efficient reallocation.  Capital
    moves from less productive to more productive hands.  Output
    rises.  This is the wealth tax as ``capitalism's
    handmaiden''---strengthening the competitive allocation mechanism
    that is capitalism's core claim to efficiency.  But the returns
    that the wealth tax preserves are $\theta_i \cdot h(E_j)$: the
    skill component is entangled with the market-inefficiency
    amplifier.
  \item Under Categories~II and~III (structural advantage and
    winner-take-all), the wealth tax accelerates concentration.
    Capital moves from investors without structural advantages to
    those with them---from small firms to dominant platforms, from new
    entrants to incumbents, from competitive sectors to oligopolistic
    ones.  The wealth tax reinforces the very market structures that
    competition policy seeks to counteract.
  \item In the typical case---a mixture of skill and structure---the
    wealth tax simultaneously rewards genuine productive ability
    \emph{and} amplifies structural rents.  The net welfare sign
    depends on the ratio $\theta_i / \gamma_j$: when skill is high
    relative to the arbitrage rate (talented entrepreneurs in
    moderately inefficient markets), the efficiency gains dominate;
    when structural persistence is extreme relative to skill
    differences (platform monopolies where the identity of the
    incumbent barely matters), the concentration costs dominate.
\end{enumerate}

The second and third cases have implications beyond efficiency.  A
competitive market economy depends on a dispersed ownership structure
where many independent actors compete, and where the price
mechanism---not the identity of the owner---determines allocation.
When the wealth tax systematically channels capital toward firms with
market power, the resulting concentration translates into political
influence, regulatory capture, and the erosion of the competitive
discipline on which the market economy depends.  The wealth tax, in
this case, does not merely fail to redistribute---it actively
concentrates economic power in positions that are already insulated
from competitive pressure.

Proposition~\ref{prop:differential} identifies the precise policy
lever.  Flow taxes compress drift differences; stock taxes preserve
them.  If drift differences arise from productive skill, compressing
them is wasteful---the Guvenen argument against income taxation.  If
drift differences arise from market power or structural incumbency,
compressing them is \emph{desirable}: it is the tax system doing what
competition policy alone cannot, by eroding the return advantage that
market structure confers.  The non-separability of skill and structure
means that any tax instrument necessarily affects both channels: a
flow tax that compresses structural rents also compresses
skill-driven returns, and a wealth tax that preserves skill-driven
returns also preserves structural rents.

The implication is that the optimal mix of flow and stock taxation
depends on an empirical prior: the share of observed return
persistence attributable to each component
of~\eqref{eq:z_interaction}.  In sectors where the
skill-times-inefficiency interaction dominates---traditional
entrepreneurship, professional services, early-stage venture
capital---the wealth tax is the superior instrument despite its
simultaneous preservation of market-inefficiency rents.  In sectors
where pure structural and platform rents dominate---platform markets,
natural monopolies, industries with strong network effects---flow
taxation is preferred.  A uniform wealth tax applied across all
sectors simultaneously gets the welfare sign right in some market
segments and wrong in others.

The framework thus points toward a more nuanced instrument design
than a simple uniform wealth tax: either sector-differentiated rates,
or a complementary combination of wealth taxation (to exploit the
skill channel) and competition policy or excess-profit taxes (to
address structural rents).  The quantitative balance depends on the
decomposition of return persistence in each economy---a question that
the Norwegian and US administrative data are, in principle, rich
enough to approximate even if a clean separation is impossible.

\subsection{Open questions}

Several directions remain open:

\paragraph{Convergence dynamics and the spectral gap.}
\citet{GabaixEtAl2016} prove that, for homogeneous random growth
processes with stabilising forces, the wealth distribution converges
to its Pareto steady state at a rate governed by the spectral gap of
the Kolmogorov forward operator.  For the basic drift--diffusion
process, the spectral gap is $\lambda_2 = \mu^2/(2\sigma^2)$, and
convergence half-lives are typically decades to centuries for
empirically plausible parameters.  This slow convergence is central to
P3's and P4's analysis of wealth tax dynamics.

The joint Fokker--Planck operator~\eqref{eq:fp_joint} is a
two-dimensional linear operator on $(x, z)$ space.  Its spectral gap
is bounded above by the minimum of two components: the
wealth-dimension gap (which, following \citeauthor{GabaixEtAl2016}'s
analysis, depends on the drift-to-diffusion ratio $v(z)/D(z)$ for the
relevant ability type) and the ability-dimension gap (which depends on
the mean-reversion rate~$\gamma$).  When ability is highly persistent
($\gamma$ small), the ability dimension contributes a near-zero
eigenvalue, making convergence even slower than the already-slow
homogeneous case.

Moreover, \citeauthor{GabaixEtAl2016}'s type-dependent analysis
(their Section~5) shows that when agents have different persistent
growth rates~$\mu_j$, the convergence rate is stratified by type:
each type~$j$ converges at its own rate
$\Lambda_j(\xi) = \xi\mu_j - \xi^2\sigma_j^2/2 + \delta$, and the
overall convergence is dominated by the slowest type.  Translated to
the continuous-ability setting of our framework: the joint distribution
converges at a rate determined by the ability type with the smallest
spectral gap.  Low-ability investors, whose drift $v(z_L)$ is small,
converge most slowly---and it is precisely these investors whose
wealth dynamics are most altered by the tax.

A further subtlety is the separation of timescales: the Pareto
exponent of the upper tail (a local property of high-$z$ investors)
adjusts relatively quickly to the wealth tax, while the full
distribution---reflecting the slow reallocation from low-$z$ to
high-$z$ types---adjusts at the rate of the joint spectral gap.
The distributional signature of a wealth tax may therefore appear in
summary statistics long before it manifests in the full distribution.

\paragraph{Spectral portfolio theory and aggregate allocations.}  The
spectral portfolio theory of \citet{Froeseth2026X} establishes that
isotropic perturbations to portfolio allocation matrices---perturbations
that affect all assets uniformly---preserve the spectral structure of
the allocation: the tail exponent of the singular value distribution,
the eigenportfolio directions, and the effective spectral rank are all
invariant.  A proportional wealth tax is exactly such an isotropic
perturbation, and the tax neutrality corollary follows.

Under heterogeneous ability, individual-level spectral invariance
still holds, but the population-level spectral structure---the
wealth-weighted aggregate of all investors' allocation
matrices---changes as $f(x,z,t)$ evolves under~\eqref{eq:fp_joint_tax}.
High-ability investors retain their portfolio allocations while
low-ability investors lose weight in the aggregate, tilting the
aggregate eigenportfolios toward the factor exposures of the able.
The research direction is to extend P7's spectral invariance from a
single-agent result to weighted mixtures of allocation matrices whose
mixing weights evolve according to the joint Fokker--Planck equation.

\paragraph{Mean-field feedback.}  In general equilibrium, asset prices
depend on the wealth distribution, creating a self-consistent
Fokker--Planck equation as in the McKean--Vlasov extension of
\citet{Froeseth2026R}.  With heterogeneous ability, the feedback
structure is richer: the composition of marginal investors affects
prices, which in turn affect the accumulation dynamics.

\paragraph{Optimal taxation.}  The optimal drift design framework of
\citet{Froeseth2026R} seeks the drift modification that minimises the
distance to a target distribution subject to intervention costs.  With
heterogeneous ability, the designer has a richer action space: taxes
can be conditioned on observable proxies for ability (fund
performance, firm profits) in addition to wealth.  \citet{SaezStantcheva2018}
derive sufficient-statistics formulas for optimal capital income
taxation that extend naturally to this setting.
\citet{GerritsenEtAl2025} show that a strictly positive tax on
capital income is Pareto-efficient whenever rates of return are
heterogeneous---whether through ability or scale effects---providing
a normative counterpart to the positive analysis here.  The ``flat
wealth tax'' result of \citet{GuvenEtAl2023}---that a uniform rate
dominates ability-contingent rates---and the broader survey of wealth
tax design in \citet{ScheuerSlemrod2021} suggest robustness properties
that warrant formal investigation within the Fokker--Planck framework.

\paragraph{Quantitative calibration.}  The perturbation parameter
controlling the deviation from neutrality---effectively the variance
and persistence of heterogeneous returns---can in principle be
estimated from the Fagereng et al.\ data.  This would provide a
quantitative measure of how far the homogeneous framework lies from
the empirical truth, broken down by asset class and market segment.

\bibliographystyle{apalike}

\end{document}